\documentclass[12pt]{article}

\usepackage{jheppub} 

\usepackage{amsmath,amsfonts,amsthm,bbold}
\usepackage[T1]{fontenc} 

\usepackage{booktabs}
\usepackage{mymacros}
\usepackage{makecell}

\newcommand\tra{^{\mathpalette\raiseT\intercal}}

\newcommand\raiseT[2]{%
\setbox0\hbox{$#1{#2}$}\raise\dp0\box0}

\makeatletter
\@addtoreset{equation}{section}
\makeatother

\newtheorem{theorem}{Theorem}
\newtheorem{lemma}{Lemma}

\renewcommand{\vec}{\mathbf}

\renewcommand{\c}{\vec c}

\newcommand{\eps}{\epsilon}
\begin{document}

\title{Inequalities for Standard Model Yukawa Couplings}
\author[a]{Gero von Gersdorff}
\author[a]{and Lucas  Modesto}
\affiliation[a]{Pontifícia Universidade Católica do Rio de Janeiro\\
Rua Marquês de São Vicente 225, Rio de Janeiro, Brazil}
\emailAdd{gersdorff@puc-rio.br}
\emailAdd{lucasmodesto@aluno.puc-rio.br}

\abstract{
We show that the Standard Model Yukawa matrices satisfy a set of simple yet nontrivial inequalities. The relations we derive are independent of the basis used to define the fermion fields, and, amongst other things, place strong constraints on the alignment of the columns of the up-type and down-type Yukawa matrices, as well as their cofactor matrices.  
The reason why one can obtain such strong statements can be traced back to the hierarchical nature of the fermion masses and quark mixings.
The inequalities should be seen as very strong necessary conditions on the SM Yukawa couplings and thus are useful for constraining   flavor models.}

\maketitle

\section{Introduction}

The very hierarchical structure of fermion masses and mixings is one of the unresolved mysteries of the Standard Model (SM).  Disregarding the neutrino sector, the properties of the fermions are encoded in their Yukawa couplings, complex three by three matrices $Y_u$, $Y_d$ and $Y_e$.
In particular, the singular values of these matrices are proportional to the masses, and the unitary Cabibbo-Kobayashi-Maskawa (CKM) matrix appearing in the charged current interactions of the quarks is a product of the basis rotations of the left handed quark fields that diagonalize the Yukawa couplings $Y_u$ and $Y_d$.
Both the quark masses and the CKM matrix elements display a strong hierarchy as summarized in table \ref{tab:data}. 

\begin{table}[h]
\centering
\begin{tabular}{@{}cc c cc c cc c|c cc c cc@{}}
\toprule
$y_u$	& $6.3\times 10^{-6}$	&&$y_d$	& $1.4 \times 10^{-5}$	&&$\theta_{12}$	&$0.23$			&&&
	$y_e$	&$2.8 \times 10^{-6}$	\\
$y_c$	& $3.1\times 10^{-3}$	&& $y_s$	& $2.7 \times 10^{-4}$ 	&& $\theta_{23}$	& $4.2\times 10^{-2}$&&&
	$y_\mu$	&$6.0 \times 10^{-4}$\\
$y_t$	& $0.87$				&& $y_b$	& $1.4 \times 10^{-2}$	&& $\theta_{13}$	& $3.7\times 10^{-3}$&&&	
	$y_\tau$	& $1.0 \times 10^{-2}$\\
\bottomrule
\end{tabular}
\caption{Quark and charged lepton data at 1 TeV in the SM \cite{Antusch:2013jca}. We assume the standard parametrization for the CKM matrix \cite{ParticleDataGroup:2024cfk}.}\label{tab:data}
\end{table}

Over the years, many approaches have been explored to account for this unusual pattern, including Abelian 
 \cite{Froggatt:1978nt}  and non-Abelian \cite{Pakvasa:1977in,Wilczek:1977uh,Barbieri:1995uv,Barbieri:1997tu} flavor symmetries,
 Yukawa textures \cite{Weinberg:1977hb,Fritzsch:1977za},
extra dimensions \cite{Grossman:1999ra,Gherghetta:2000qt,Huber:2000ie,Cabrer:2011qb,vonGersdorff:2013rwa}, 
and clockwork models \cite{Burdman:2012sb,vonGersdorff:2017iym,Patel:2017pct,Alonso:2018bcg,Smolkovic:2019jow,AbreudeSouza:2019ixc,vonGersdorff:2020ods}.

The Yukawa matrices contain a lot of unmeasureable information, related to the fact that the remaining part of the SM Lagrangian is invariant under a large  $\mathrm U(3)^5$ global flavor symmetry that allows one to absorb most of the unitary rotations  diagonalizing the Yukawa matrices.
Therefore, rather than fixing the Yukawa matrices completely, the physical data of table \ref{tab:data}  impose some highly complicated constraints on them (for instance, the singular values are the roots of the characteristic polynomial of $Y^\dagger Y$, which depends in a rather complex way on the $Y_{ij}$).
From a model builder’s perspective, these constraints are not particularly practical -- neither for devising a mechanism behind the peculiar hierarchical structure nor for performing the actual analytical diagonalization needed to fix the model parameters.
The question that we would like to pose in this work is: can we derive some simpler (albeit only necessary) conditions on the Yukawa matrices that any model of flavor must obey in order to reproduce the data in table \ref{tab:data}? The answer to this question is yes, and
these conditions take the form of inequalities that are much simpler than the aforementioned equalities. By 
moving from equalities to inequalities we give up on sufficiency but gain a lot in simplicity. At the same time, some of the inequalities are incredibly sharp, leaving only very small windows (often well below one percent). The technical reason behind this is again the hierarchical nature of the masses and mixings.

\section{Derivation of the Inequalities}
\label{sec:inequalities}

\subsection{Approximation for Singular Value Decompositions}

The basis for the inequalities is a very general approximation scheme  for the singular value 
decomposition (SVD) of arbitrary complex $3\times 3$ matrices \cite{vonGersdorff:2019gle}. 
The approximate expressions for the singular values and the unitary rotations must fall in a certain computable interval containing the exact values, and this interval turns out to be quite small due to the hierarchical structure of the data. 
We will be working in the lowest order approximation (order $n=\frac{1}{2}$ in the language of reference \cite{vonGersdorff:2019gle}), which strikes the best balance between simplicity and accuracy.

Let us start by introducing some notation. We denote the columns  of the up-quark Yukawa matrix by $\c_{u_i}$, i.e.
\be
Y_u=(\c_{u_1}\  \c_{u_2}\  \c_{u_3})\,,
\ee
and similarly for $Y_d$ and $Y_e$.
Moreover, let us denote by $\tilde Y_u$ the cofactor matrix of $Y_u$,
\be
\tilde Y_u\equiv (\c_{u_2}\!\times \c_{u_3} \ \c_{u_3}\!\times \c_{u_1} \ \c_{u_1}\!\times \c_{u_2} )\,.
\ee
The cofactor matrix satisfies $\tilde Y_u\tra Y^{}_u=\det Y_u^{}\one$. We further note that $Y_u$ and 
$\tilde Y_u^*$ are diagonal in the same basis.
Their  SVD's read 
\begin{align}
Y_u^{}={}&U_{u_L}^{}
\mathcal Y_u^{}\label{eq:svd}
\, U_{u_R}^\dagger\,, \\
\tilde Y_u^{}={}&U_{u_L}^{*}
\tilde{\mathcal  Y}_u^{}
\, U_{u_R}\tra e^{i\arg\det Y_u}\,,
\label{eq:svdtilde}
\end{align}
with $\mathcal Y_u=\diag(y_u,y_c,y_t)$ and $\tilde {\mathcal Y}_u=\diag(y_t y_c, y_t y_u, y_c y_u)$ containing the physical Yukawa couplings from table \ref{tab:data} or the respective values at any other scale,
and $U_{u_L}$, $U_{u_R}$ are unitary matrices.

\newcommand{\vtl}{\vec  v_{t_L}}
\newcommand{\vul}{\vec  v_{u_L}}
\newcommand{\vtr}{\vec  v_{t_R}}
\newcommand{\vur}{\vec  v_{u_R}}

\newcommand{\utl}{\vec  u_{t_L}}
\newcommand{\uul}{\vec  u_{u_L}}
\newcommand{\utr}{\vec  u_{t_R}}
\newcommand{\uur}{\vec  u_{u_R}}
\newcommand{\ucl}{\vec  u_{c_L}}
\newcommand{\ucr}{\vec  u_{c_R}}

\newcommand{\vbl}{\vec  v_{b_L}}
\newcommand{\vdl}{\vec  v_{d_L}}
\newcommand{\vbr}{\vec  v_{b_R}}
\newcommand{\vdr}{\vec  v_{d_R}}

\newcommand{\vtaul}{\vec  v_{\tau_L}}
\newcommand{\vtaur}{\vec  v_{\tau_R}}
\newcommand{\vel}{\vec  v_{e_L}}
\newcommand{\ver}{\vec  v_{e_R}}

\newcommand{\ubl}{\vec  u_{b_L}}
\newcommand{\udl}{\vec  u_{d_L}}
\newcommand{\ubr}{\vec  u_{b_R}}
\newcommand{\udr}{\vec  u_{d_R}}
\newcommand{\usl}{\vec  u_{s_L}}
\newcommand{\usr}{\vec  u_{s_R}}

\begin{table}
\centering
\begin{tabular}{cccc}
\toprule
Matrix 	& Longest column 		& Longest row 		& Common element 	\\
\hline
$Y_u$	  	&$\vtl$	&$\vtr^\dagger$	& $y_{tt}$		\\
$\tilde Y_u$	&$\vul^*$	&$\vur\tra$	& $\tilde y_{uu}$\\
\midrule
$Y_d$	  	&$\vbl$	&$\vbr^\dagger$	& $y_{bb}$		\\
$\tilde Y_d$	&$\vdl^*$	&$\vdr\tra$	& $\tilde y_{dd}$\\
\midrule
$Y_e$	  	&$\vtaul$	&$\vtaur^\dagger$	& $y_{\tau\tau}$		\\
$\tilde Y_e$	&$\vel^*$	&$\ver\tra$	& $\tilde y_{ee}$\\
\bottomrule
\end{tabular}
\caption{Notation used for the longest columns and rows of $Y_{u,d,e}$ and $\tilde Y_{u,d,e}$.}\label{tab:defs}
\end{table}

We now  define $\vtl $ to be the longest column vector of $Y_u$, in other words $\vtl$ satisfies
\be
|\vtl |=\max_i |\c_{u_i}|\,.
\ee 
Similarly we define $\vul$ to be the longest column of  $\tilde Y_u^{*}$, $\vtr^\dagger$  the longest row of $Y_u$, and $\vur\tra$ the longest row of $\tilde Y_u$. Finally, the unique component of $Y_u$ ($\tilde Y_u$) that is common to both $\vtl$ and $\vtr^\dagger$ ($\vul^*$ and $\vur\tra$) is denoted by $y_{tt}$ ($\tilde y_{uu}$).
For the down and charged lepton sector we repeat these definitions, they are  summarized in table \ref{tab:defs}. 
Finally,  we define the unit vectors of   table \ref{tab:orthonormal-left}.\footnote{Analogous definitions apply for the right-handed vectors and for the lepton sector, but they will not appear in our main inequalities, so we omit them from table \ref{tab:defs} for clarity.}

\begin{table}
\centering
\begin{tabular}{cc}
\toprule
Sector &  Unit vectors\\
\midrule
$u$ &
$\utl \equiv \vtl/|\vtl|\quad
\uul \equiv \vul/|\vul|\quad
\ucl \equiv \utl^*\times\uul^*$ \\
$d$ &
$\ubl \equiv \vbl/|\vbl|\quad
\udl \equiv \vdl/|\vdl|\quad
\usl \equiv \ubl^*\times\udl^*$ \\
\bottomrule
\end{tabular}
\caption{Notation used for the normalized longest columns of $Y_{u,d}$ and $\tilde Y_{u,d}$.}
\label{tab:orthonormal-left}
\end{table}

Let us, for definiteness, focus on the up sector. As shown in ref.~\cite{vonGersdorff:2019gle} the matrices
\be
 U'_{u_L}\equiv (\uul,\ucl,\utl)\,,
\ee
\be
 U'_{u_R}\equiv (\uur,\ucr,\utr)\,,
\ee
 approximate the  transformations $U_{u_L}$, $U_{u_R}$ of the left-handed quark fields to the mass eigenbasis. 
In order to quantify the deviation of $U_x'$ from $U^{}_x$ we define 
\be
V_{u_L}\equiv  U_{u_L}^\dagger  U'^{}_{u_L}\,,
\label{eq:Vx}
\qquad
V_{u_R}\equiv  U_{u_R}^\dagger  U'^{}_{u_R}\,.
\ee
One can then show that the following inequalities hold: 
\be
\begin{pmatrix}
1-|V_{11}|^2	&	|V_{12}|^2	&	|V_{13}|^2\\
|V_{21}|^2		&	1-|V_{22}|^2&	|V_{23}|^2\\
|V_{31}|^2		&	|V_{32}|^2	&	1-|V_{33}|^2
\end{pmatrix}
\leq \begin{pmatrix}
Z_u(\frac{y_u}{y_c})^{2}	& 	Z'_u(\frac{y_u}{y_c})^{2}	&	X_u(\frac{y_u}{y_t})^{2}\\
Z_u(\frac{y_u}{y_c})^{2}	&	Z'_u(\frac{y_u}{y_c})^{2}+X'_u(\frac{y_c}{y_t})^{2}	&	X_u(\frac{y_c}{y_t})^{2}\\
Z_u(\frac{y_u}{y_t})^{2}	&	X'_u(\frac{y_c}{y_t})^{2}	&	X_u(\frac{y_c}{y_t})^{2}
\end{pmatrix}\,,
\label{eq:master}
\ee
where $V=V_{u_L}$ or $V_{u_R}$, and $X_u$ and $Z_u$ are defined as
\be
X_u\equiv \frac{2y_t^2-y_c^2}{y_t^2-2y_c^2}\,,
\qquad
Z_u\equiv \frac{2y_c^2-y_u^{2}}{y_c^{2}-2y_u^{2}}\,,
\ee
\be
X_u'\equiv \frac{y_t^2}{y_t^2-9y_u^2}X_u\,,
\qquad
Z_u'\equiv \frac{y_t^2}{y_t^2-9y_u^2}Z_u\,.
\ee
which are numerically very close to two.\footnote{With the central values of table \ref{tab:data},
they  range from $Z_u-2=1.2\times 10^{-5}$ to $X_e-2=0.011$.
} 
Analogous relations hold for the down and lepton sectors. The inequalities in eq.~(\ref{eq:master}) are exact, no approximations have been made.
Notice that the  components of the matrix on the right hand side are strongly suppressed, indicating that $U'$ approximates $U$ rather well.

We provide details of the derivation of eq.~(\ref{eq:master})  in appendix \ref{sec:svd}. \footnote{A relation very similar to 
eq.~(\ref{eq:master}) was already derived in \cite{vonGersdorff:2019gle} under the assumption that $U'$ and hence $V$ are exactly unitary, which corresponds to the case $\alpha\neq\gamma$, $\beta\neq\delta$ in the language of appendix \ref{sec:svd}. 
In contrast, eq~(\ref{eq:master}) is valid without this assumption.
}
It makes essential use of the longest column criterion, whose importance can be intuitively understood as follows.  Neglecting all but $y_t$, the matrix $Y^{}_u$ has rank one, and hence 
$Y_u^{}Y_u^\dagger Y^{}_u=y_t^2 Y_u^{}$. Therefore, any of the three columns of $Y_u$ (if nonzero) is an  eigenvector of $Y_u^{}Y_u^\dagger$ to eigenvalue $y_t^2$. 
This is still approximately true if we turn on $y_c$ and $y_u$, as they only provide small corrections to $Y_u$. However if the entries in the chosen column of $Y_u$ in the limit
 $y_c,y_u\to 0$ are small, the corrections are potentially important and the approximation breaks down. This pathological case is avoided by selecting the longest column of $Y_u$, which always remains a rather accurate approximation for the actual eigenvector.

Moreover, the SVD in the form $(Y_uU_{u_R})_{i3}=(U_{u_L})_{i3}y_t$ provides
 three different approximations of $y_t$ (one for each value of $i$) 
\be
y_t=\frac{(Y^{}_uU_{u_R})_{i3}}{(U_{u_L})_{i3}}\approx\frac{(Y^{}_uU'_{u_R})_{i3}}{(U'_{u_L})_{i3}}
\label{eq:ytprime0}
\ee
The most robust approximation (avoiding the pathological case $(U'_{u_L})_{i3}=0$) is obtained by chosing 
$i$ as the index of the longest row of $Y_u$. The right hand side of eq.~(\ref{eq:ytprime0}) can then be rewritten in the form 
\be
 y'_t\equiv \frac{|\vtl ||\vtr|}{|y_{tt}|}\,.
\label{eq:ytprime}
\ee

We can derive inequalities analogous to eq.~(\ref{eq:master}) to bound the errors of $ y'_t$. They 
 are given by  (see appendix \ref{sec:svd} for details of the derivation)
\be
\frac{y_t^2}{y_t+y_c X_u}\leq \frac{|\vtl ||\vtr|}{|y_{tt}|}\leq \frac{y_t^2+y_c^2X_u}{y_t-y_cX_u}\,.
\label{eq:yt}
\ee
A similar relation holds for the SVD of the cofactor matrix 
\be
\frac{(y_cy_t)^{2}}{y_c y_t+y_u y_tZ_u}
\leq
\frac{|\vul||\vur|}{|\tilde y_{uu}|}\leq 
\frac{(y_cy_t)^{2}+(y_u y_t)^{2}Z_u}{y_c y_t-y_u y_tZ_u}\,.
\label{eq:ycyt}
\ee
The relations eq.~(\ref{eq:yt}) and eq.~(\ref{eq:ycyt}) as well as their analogues  for the down and charged lepton sectors give the relations in the upper half of table \ref{tab:3}.
The quotient of the difference of the limits over their sum ranges from 0.4 to 10 \%.

These inequalities  are somewhat nontrivial: from $\tr Y_u^\dagger Y_u^{}=y_t^2+y_c^2+y_u^2$ one concludes that 
 $|\vtl|$, $|\vtr|$  and $y_{tt}$, in the limit of $y_u, y_c\to 0$, satisfy the much weaker separate inequalities $\frac{1}{\sqrt{3}}y_t\leq|\vtl|\leq y_t$, $\frac{1}{\sqrt{3}}y_t\leq|\vtr|\leq y_t$ 
 and $\frac{1}{3}y_t\leq |y_{tt}|\leq y_t $, these bound receive some minor corrections for nonzero $y_{u,c}$.
The particular combination eq.~(\ref{eq:yt}) is however much more tightly constrained.

\begin{table}
\centering
\begin{tabular}{cccc}
\toprule
\makecell{approximated\\ quantity}	& expression	 			              & lower bound 	         & upper bound\\
\midrule
$y_t$					& $\frac{|\vtl ||\vtr|}{|y_{tt}|}$			    & $8.638 \times 10^{-1}$	& $8.762 \times 10^{-1}$	\\
$y_c y_t$				& $\frac{|\vul ||\vur|}{|\tilde y_{uu}|}$	    & $2.686 \times 10^{-3}$	& $2.708 \times 10^{-3}$	\\
$y_b$					& $\frac{|\vbl ||\vbr|}{|y_{bb}|}$			    & $1.345 \times 10^{-2}$	& $1.460 \times 10^{-2}$	\\
$y_s y_b$				& $\frac{|\vdl ||\vdr|}{|\tilde y_{dd}|}$	    & $3.416 \times 10^{-6}$	& $4.252 \times 10^{-6}$	\\
$y_\tau$				& $\frac{|\vtaul ||\vtaur|}{|y_{\tau\tau}|}$	& $8.917 \times 10^{-3}$	& $1.146 \times 10^{-2}$	\\
$y_\mu y_\tau$			& $\frac{|\vel ||\ver|}{|\tilde y_{ee}|}$		& $5.941 \times 10^{-6}$	& $6.060 \times 10^{-6}$	\\
\midrule
$|V_{ud}|$			    & $|\uul^\dagger \udl^{}|$	                & $9.541 \times 10^{-1}$	& $9.916 \times 10^{-1}$	\\
$|V_{us}|$			    & $|\uul^\dagger \usl^{}|$	                & $1.501 \times 10^{-1}$	& $2.999 \times 10^{-1}$	\\
$|V_{ub}|$			    & $|\uul^\dagger \ubl^{}|$	                & $0$			            & $1.145 \times 10^{-2}$	\\
$|V_{cd}|$			    & $|\ucl^\dagger \udl^{}|$	                & $1.501 \times 10^{-1}$	& $2.997 \times 10^{-1}$	\\
$|V_{cs}|$			    & $|\ucl^\dagger \usl^{}|$	                & $9.514 \times 10^{-1}$	& $9.922 \times 10^{-1}$	\\
$|V_{cb}|$			    & $|\ucl^\dagger \ubl^{}|$	                & $1.098 \times 10^{-2}$	& $7.489 \times 10^{-2}$	\\
$|V_{td}|$			    & $|\utl^\dagger \udl^{}|$	                & $3.009 \times 10^{-3}$	& $1.507 \times 10^{-2}$	\\
$|V_{ts}|$			    & $|\utl^\dagger \usl^{}|$	                & $9.097 \times 10^{-3}$	& $7.506 \times 10^{-2}$	\\
$|V_{tb}|$			    & $|\utl^\dagger \ubl^{}|$	                & $9.971 \times 10^{-1}$	& $1$		\\
\bottomrule
\end{tabular}
\caption{Summary of the inequalities.
Under the assumption that $Y_u$, $Y_d$, and $Y_e$ exactly reproduce the data of table \ref{tab:data}, the inequalities
(lower bound) $\leq$ expression $\leq$ (upper bound) hold.
}
\label{tab:3}
\end{table}

\subsection{Inequalities from the CKM Mixing Angles}

Starting from the approximations for the left handed normalized eigenvectors of $Y_u$ and $Y_d$ (see table \ref{tab:orthonormal-left}), we can build an
 approximate CKM matrix:
\be
 V'_{\rm CKM}\equiv 
\begin{pmatrix}
\uul^\dagger \udl^{} & \uul^\dagger \usl^{} & \uul^\dagger \ubl^{} \\
\ucl^\dagger \udl^{} &\ucl^\dagger \usl^{} &\ucl^\dagger \ubl^{} \\
\utl^\dagger \udl^{} &\utl^\dagger \usl^{} &\utl^\dagger \ubl^{}
\end{pmatrix}
=V_{u_L}^\dagger  V^{}_{\rm CKM} V_{d_L}^{}
\label{eq:CKMtilde}\,,
\ee
where $V_{\rm CKM}$ is the physical CKM matrix, and $V_{u_L,d_L}$ were defined in eq.~(\ref{eq:Vx}). From  eq.~(\ref{eq:master}) we can then derive upper and lower bounds for this approximation.
Let us define
\be
\delta_u\equiv
\begin{pmatrix}
Z_u(\frac{y_u}{y_c})^{2}\\
&Z_u(\frac{y_u}{y_c})^{2}+X_u(\frac{y_c}{y_t})^{2}\\
&&X_u(\frac{y_c}{y_t})^{2}
\end{pmatrix}\,,
\qquad
\eps_u\equiv  \begin{pmatrix}
0	& 	\sqrt{Z_u}\frac{y_u}{y_c}	&	\sqrt {X_u}\frac{y_u}{y_t}\\
\sqrt {Z_u}\frac{y_u}{y_c}	&	0	&	\sqrt {X_u}\frac{y_c}{y_t}\\
\sqrt {Z_u}\frac{y_u}{y_t}	&	\sqrt {X_u} \frac{y_c}{y_t}	&	0
\end{pmatrix}\,,
\ee
and similarly for $\delta_d$ and $\eps_d$. In the following we will use the shorthand $|A|$ to denote the matrix whose components are the absolute values of the components of the matrix $A$, i.e., $|A|_{ij}=|A_{ij}|$.
Defining
\be
\Delta_{\rm CKM}\equiv \eps_u\tra |V_{\rm CKM}|+|V_{\rm CKM}|\eps_d^{}+\eps_u\tra |V_{\rm CKM}|\eps^{}_d\,,
\ee
we can easily   
 derive an upper bound\footnote{The min and max functions in eq.~(\ref{eq:ckmup}) and  (\ref{eq:ckmlow}) apply component-wise.}
 \be
| V'_{\rm CKM}|\leq 
	\min\left(1,\, |V_{\rm CKM}|+\Delta_{\rm CKM}\right)
	\equiv  V_{\rm CKM}'^+\,,
\label{eq:ckmup}
\ee
and a lower bound 
\be
| V'_{\rm CKM}|\geq \max\, (0,\,\sqrt{\one-\delta_u}|V_{\rm CKM}|\sqrt{\one-\delta_d}
	-\Delta_{\rm CKM})
\equiv  V_{\rm CKM}'^-\,.
\label{eq:ckmlow}
\ee
In deriving eqns.~(\ref{eq:ckmup}) and (\ref{eq:ckmlow})  we have used the triangle inequality, eq.~(\ref{eq:master}), as well as $0\leq |W_{ij}|\leq 1$  valid for any unitary matrix $W$.
Note that $ V'_{\rm CKM}$ is expressed in terms of the Yukawa matrices via eq.~(\ref{eq:CKMtilde}) and $V_{\rm CKM}'^\pm$ only depend on experimental data and can be computed. 
Plugging in the central values given in table \ref{tab:data}, we obtain  
 the lower half of table \ref{tab:3}.
%
%
%

One of the simplest  and at the same time most stringent relations 
 comes from  $V_{tb}$:
\be
0.9971\leq |\utl^\dagger \ubl^{}|\leq 1\,,
\label{eq:utub}
\ee
It essentially says that the longest columns of $Y_u$  and $Y_d$ have to be aligned with an accuracy of only 0.3\%.
Next, we have the relations from $V_{ub}$ and $V_{td}$ respectively
\begin{align}
0\leq |\uul^\dagger \ubl^{}|\leq 0.01145\,,\\
0.003009\leq |\utl^\dagger \udl^{}|\leq 0.01507\,,
\label{eq:utud}
\end{align}
which state that the longest columns of $Y_d$ and $\tilde Y_u^*$ (or $Y_u$ and $\tilde Y_d^*$ repectively) should be orthogonal within 1\%. Additionally, in the case of $|\utl^\dagger \udl^{}|$ there is a nontrivial lower bound as well.
Similar to eq.~(\ref{eq:utub}) we also have 
\be
0.9541\leq |\uul^\dagger \udl^{}|\leq 0.9916\,,
\label{eq:utub}
\ee
which quantifies the alignment between the longest columns of $\tilde Y_u$ and $\tilde Y_d$.
Other equations are less constrained and also slightly more complex (for instance, they involve the second-generation approximate eigenvectors, which are already cubic in the original Yukawa couplings).

\section{Numerical Checks}

In order to cross check the inequalities of the previous section 
 we have generated random Yukawa couplings as follows. In order to keep the data of table \ref{tab:data} fixed, we sampled\footnote{We draw these matrices from the invariant Haar measure for $\operatorname{ SU}(3)$.} unitary matrices $U_{u_R}$, $U_{d_L}$ and $U_{d_R}$, while setting $U_{u_L}=U_{d_L}V^\dagger_{\rm CKM}$ in order to fix the CKM matrix.
From the matrices $Y_u$ and  $Y_d$ obtained in this way we compute the quantities in the second column of table \ref{tab:3}.
The results are plotted in figures \ref{fig:svs} and \ref{fig:ckm}. We sampled the Yukawa matrices one million times and observe no violation of the inequalities.

\begin{figure}[h]
\centering
\includegraphics[width=7.5cm]{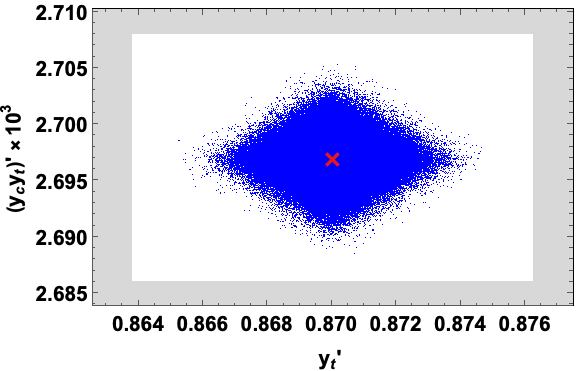}
\hfill
\includegraphics[width=7.5cm]{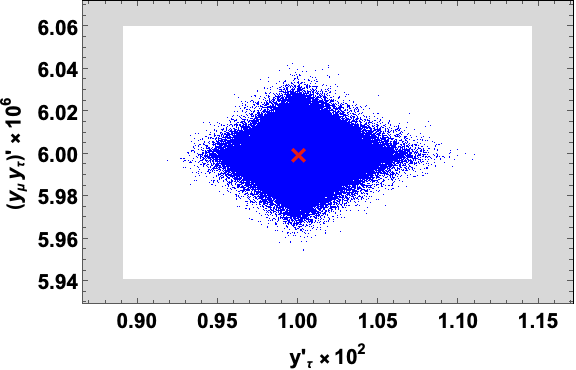}
\caption{Left panel: simulation of the inequalities in rows 1 and 2  of table \ref{tab:3}. Right panel: same for  rows 5 and 6. The gray area is  forbidden by the inequalities, and the red crosses correspond to the physical values.}
\label{fig:svs}
\end{figure}

\begin{figure}[h]
\centering
\includegraphics[width=7.5cm]{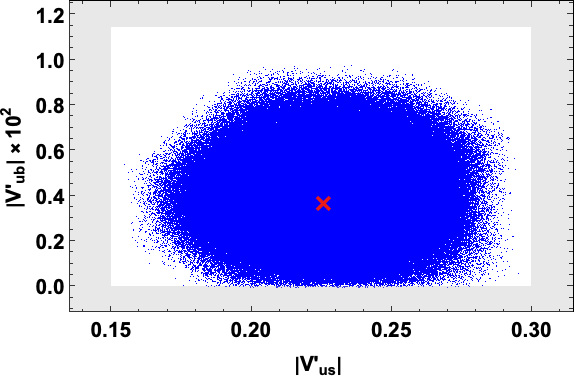}
\hfill
\includegraphics[width=7.5cm]{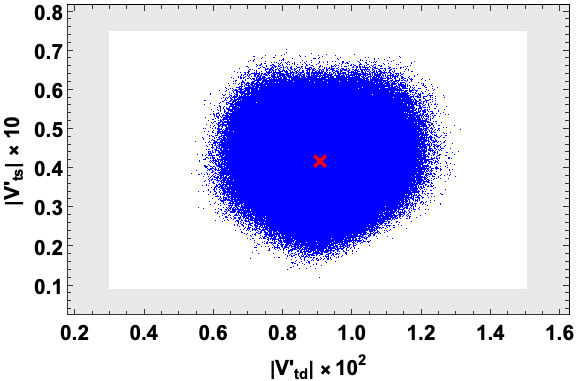}
\\
\vspace{1cm}
\includegraphics[width=7.5cm]{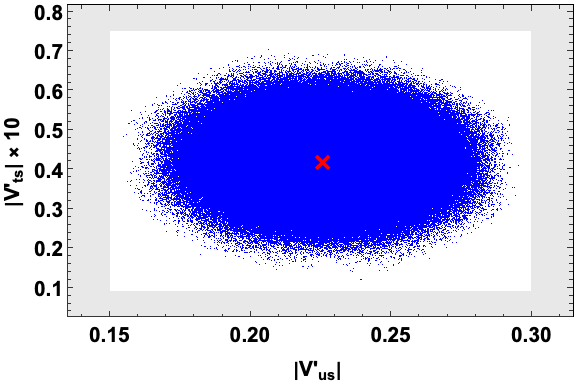}
\hfill
\includegraphics[width=7.5cm]{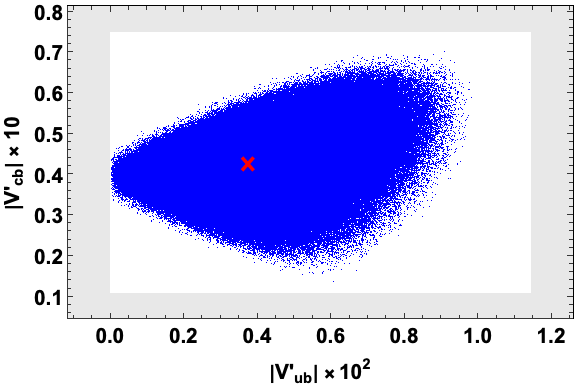}
\caption{Simulation of various inequalities in the lower part of table \ref{tab:3}. The gray area is  forbidden by the inequalities, and the red crosses correspond to the physical values.}
\label{fig:ckm}
\end{figure}

Notice that it appears from the plots that the inequalities are not completely optimal, as the bounds are not always reached. Partially, this is only a spurious effect due to the low probability densities near the bounds. For instance, in order to saturate the inequalities (\ref{eq:yt}), it is necessary that both inequalities in Lemma \ref{lem:basic} are saturated.
This in turn requires both $U_{u_L}$ and $U_{u_R}$  to be very close to 
"tribimaximality" ($|U_{u_L\,i3}|^2=\frac{1}{3}$, $|U_{u_L\,\alpha 2}|^2=\frac{2}{3}$ and $|U_{u_L\,\alpha 1}|=0$, where $i=1,2,3$ and $\alpha$ is defined in Lemma \ref{lem:basic}), and similarly for $U_{u_R}$ with $\alpha\to\beta$. We have checked that close to this point the inequalities for $y_t'$ can indeed be saturated, while even with one million simulations it is unlikely that a randomly drawn sample comes close to this point.

On the other hand, the inequalities for the CKM angles are likely not completely optimal, as this would require that the inequalities of Lemmas \ref{lem:basic} and \ref{lem:basic2} are simultaneously saturated, which is not possible. Still it does appear that only the bounds on $(V'_{\rm CKM})_{td}$ and the upper bound on $(V'_{\rm CKM})_{ub}$ can be notably improved.

\section{Conclusions}

We have derived strong inequalities for certain combinations of the SM Yukawa matrix elements that any model of flavor must necessarily obey. The inequalities are compiled in table \ref{tab:3}; see also tables \ref{tab:defs} and \ref{tab:orthonormal-left} for a summary of the notation. 

We stress once more that these relations are completely independent of the basis used to define the quark and lepton fields, they are true for any set of Yukawa matrices $Y_u$, $Y_d$, $Y_e$ reproducing the data in table \ref{tab:data}. The bounds are valid for the couplings at 1 TeV, for any other scale the corresponding intervals can be derived using the respective data in eqns.~(\ref{eq:yt}),  (\ref{eq:ycyt}),  (\ref{eq:ckmup}), and (\ref{eq:ckmlow}).

We could go to higher orders in the approximation scheme of ref.~\cite{vonGersdorff:2019gle}. This would result in bounds on columns/rows of $YY^\dagger$, $YY^\dagger Y$ etc., with corresponding substitutions of $y_i$ by $y_i^2$, $y_i^3$ etc. in all bounds. Even though the intervals  shrink extremely rapidly to the true masses and CKM mixing angles, the price to pay is that  the approximate expressions for them become more complex functions of the Yukawa couplings. We believe that the lowest order approximation  strikes the best balance between simplicity and accuracy.

The list of inequalities derived in this work is not exhaustive. There are also constraints on the second longest columns/rows and inequalities based on slightly different approximations for the singular values. 
We leave the derivation of these constraints to future work.

\section*{Acknowledgements}
GG acknowledges financial support by the Conselho Nacional de Desenvolvimento
Científico e Tecnológico (CNPq) under fellowship number 313238/2023-5, as well as 
the Funda\char"00E7ão de Amparo à Pesquisa do Estado do Rio de Janeiro (FAPERJ) under 
project number 210.785/2024. LM is supported by Coordena\char"00E7ão de Aperfei\char"00E7oamento de Pessoal de Nível Superior (CAPES).

\appendix

\section{Inequalities from Singular Value Decompositions}
\label{sec:svd}

\begin{lemma}
\label{lem:basic}
Let Y be a complex 3 by 3 matrix, and let $Y=U_L^{}\diag(y_1,y_2,y_3)U_R^\dagger$ be its SVD with $y_1\leq y_2\leq y_3$. 
Furthermore, let the $\alpha^{\rm th}$ row ($\beta^{\rm th}$ column) be the longest row (column) of $Y$.  Then the following inequalities hold:
\be
|U_{L\alpha 3}|^2\geq  \frac{1}{X+1}\,,\qquad \frac{|U_{L\alpha 1}|^2}{|U_{L\alpha 3}|^2}+\frac{|U_{L\alpha 2}|^2}{|U_{L\alpha 3}|^2}\leq X\,,
\label{eq:basicL3}
\ee
\be
|U_{R\beta 3}|^2\geq \frac{1}{X+1} \,, \qquad \frac{|U_{R\beta 1}|^2}{|U_{R\beta 3}|^2}+\frac{|U_{R\beta 2}|^2}{|U_{R\beta 3}|^2}\leq X\,,
\label{eq:basicR3}
\ee
where $X\equiv\frac{2y_3^2-y_2^2}{y_3^2-2y_2^2}$.\footnote{We also define $X\equiv\infty$ if $y_3^2< 2 y_2^2$.}
\end{lemma}
\begin{proof}
The length squared of the $j^{\rm th}$ row  is given by $(Y Y^\dagger)_{jj}$. The lengths squared of the three rows sum to $\tr YY^\dagger=y_1^2+y_2^2+y_3^2$, and hence 
\be
\frac{1}{3}(y_3^2+y_2^2+y_1^2)\leq (Y Y^\dagger)_{\alpha\alpha}= |U^{}_{L\,\alpha i}|^2y_i^2\leq
\left(|U_{L\,\alpha 1}^{}|^2+|U_{L\,\alpha 2}^{}|^2\right)y_2^2+|U_{L\,\alpha 3}^{}|^2y_3^2\,,
\ee
where in the last step we used $y_1\leq y_2$.
Unitarity of the $\alpha^{\rm th}$ row of $U_L$ then proves eq.(\ref{eq:basicL3}), 
\be
|U_{L\,\alpha 3}|^2\geq \frac{y_3^2-2y_2^2+y_1^2}{3(y_3^2-y_2^2)}\geq \frac{y_3^2-2y_2^2}{3(y_3^2-y_2^2)}=\frac{1}{X+1}\,.
\ee
 The steps to prove eq.~\ref{eq:basicR3} are identical.
\end{proof}
These inequalities are valid for any singular values, but they become the most restrictive in the case of strongly hierarchical ones, in which case
$X\approx 2$. There is an analogous Lemma for the longest rows and columns of the cofactor matrix:

\begin{lemma}
\label{lem:basic2}
Let Y be a complex 3 by 3 matrix,  let $Y=U_L^{}\diag(y_1,y_2,y_3)U_R^\dagger$ be its SVD with $y_1\leq y_2\leq y_3$, and let $\tilde Y$ be the cofactor matrix of $Y$. Furthermore, let $\gamma$ and $\delta$ index the longest row and column of $\tilde Y$ respectively. Then the following inequalities hold:
\be
|U_{L\gamma 1}|^2\geq  \frac{1}{Z+1}\,,\qquad \frac{|U_{L\gamma 2}|^2}{|U_{L\gamma 1}|^2}+\frac{|U_{L\gamma 3}|^2}{|U_{L\gamma 1}|^2}\leq Z\,,
\label{eq:basicL3tilde}
\ee
\be
|U_{R\delta 1}|^2\geq \frac{1}{Z+1} \,, \qquad \frac{|U_{R\delta 2}|^2}{|U_{R\delta 1}|^2}+\frac{|U_{R\delta 3}|^2}{|U_{R\delta 1}|^2}\leq Z\,,
\label{eq:basicR3tilde}
\ee
where $Z\equiv \frac{2y_2^{2}-y_1^2}{y_2^2-2y_1^2}$.
\end{lemma}
\begin{proof}
The proof is identical to the one of Lemma \ref{lem:basic}, making use of the SVD of $\tilde Y$ in eq.~(\ref{eq:svdtilde}).
\end{proof}

\begin{theorem}
\label{thm:sv3}
Let Y be a complex 3 by 3 matrix, $Y=U_L^{}\diag(y_1,y_2,y_3)U_R^\dagger$ be its SVD with $y_1\leq y_2\leq y_3$, 
and let the $\alpha^{\rm th}$ row ($\beta^{\rm th}$ column) be the longest row (column) of $Y$.  Then the following inequalities hold
\be
\frac{y_3^2}{y_3+Xy_2}\leq \frac{[(YY^\dagger)_{\alpha\alpha}(Y^\dagger Y)_{\beta\beta}]^\frac{1}{2}}{|Y_{\alpha\beta}|}
\leq \frac{y_3^2+Xy_2^2}{y_3-Xy_2}\,,
\label{eq:thmyt}
\ee
where $X\equiv\frac{2y_3^2-y_2^2}{y_3^2-2y_2^2}$.
\end{theorem}
\begin{proof}
Writing the expression in eq.~(\ref{eq:thmyt}) in terms of the SVD of $Y$
\be
\frac{[(YY^\dagger)_{\alpha\alpha}(Y^\dagger Y)_{\beta\beta}]}{|Y_{\alpha\beta}|^2}
=y_3^2\frac{\left(1+\frac{| U_{L\,\alpha 2}|^2}{|U_{L\,\alpha 3}|^2}\frac{y_2^2}{y_3^2}+\frac{| U_{L\,\alpha 1}|^2}{|U_{L\,\alpha 3}|^2}\frac{y_1^2}{y_3^2}\right)\left(1+\frac{| U_{R\,\beta 2}|^2}{|U_{R\,\beta 3}|^2}\frac{y_2^2}{y_3^2}+\frac{| U_{R\,\beta 1}|^2}{|U_{R\,\beta 3}|^2}\frac{y_1^2}{y_3^2}\right)}
{\left|
1+\frac{U_{L\, \alpha 2}U^*_{R\,\beta 2}}{U_{L\, \alpha 3}U^*_{R\,\beta 3}}\frac{y_2}{y_3}
 +\frac{U_{L\, \alpha 1}U^*_{R\,\beta 1}}{U_{L\, \alpha 3}U^*_{R\,\beta 3}}\frac{y_1}{y_3}\right|^2}\,,
\ee
we can straightforwardly apply  Lemma \ref{lem:basic} to get the following bounds for the numerator terms
\be
1\leq 
1+\frac{| U_{L\,\alpha 2}|^2}{|U_{L\,\alpha 3}|^2}\frac{y_2^2}{y_3^2}\
+\frac{| U_{L\,\alpha 1}|^2}{|U_{L\,\alpha 3}|^2}\frac{y_1^2}{y_3^2}\\
\leq
1+\frac{y_2^2}{y_3^2}X\,,
\ee
\be
1\leq
1+\frac{| U_{R\,\beta 2}|^2}{|U_{R\,\beta 3}|^2}\frac{y_2^2}{y_3^2}+\frac{| U_{R\,\beta 1}|^2}{|U_{R\,\beta 3}|^2}\frac{y_1^2}{y_3^2}
\leq
1+\frac{y_2^2}{y_3^2}X\,.
\ee
For the  denominator we use the triangle inequalities ($\left||a|-|b|\right|\leq|a+b|\leq |a|+|b|$ for any complex numbers $a,b$) and maximize/minimize the resulting expressions under the conditions eq.~(\ref{eq:basicL3}) and eq.~(\ref{eq:basicR3}) to get
\be
1-\frac{y_2}{y_3}X
\leq
\left|
1+\frac{U^{}_{L\, \alpha 2}U^*_{R\,\beta 2}}{U^{}_{L\, \alpha 3}U^*_{R\,\beta 3}}\frac{y_2}{y_3}
 +\frac{U^{}_{L\, \alpha 1}U^*_{R\,\beta 1}}{U^{}_{L\, \alpha 3}U^*_{R\,\beta 3}}\frac{y_1}{y_3}\right|
 \leq 1+\frac{y_2}{y_3}X\,.
\ee
Combining all these inequalities we straightforwardly arrive at eq.~(\ref{eq:thmyt}).
\end{proof}

The analogous Theorem for the cofactor matrix reads
\begin{theorem}
Let Y be a complex 3 by 3 matrix,  let $Y=U_L^{}\diag(y_1,y_2,y_3)U_R^\dagger$ be its SVD with $y_1\leq y_2\leq y_3$, and let $\tilde Y$ be the cofactor matrix of $Y$. Furthermore, let $\gamma$ and $\delta$ index the longest row and column of $\tilde Y$ respectively. 
Then the following inequalities hold:
\be
\frac{y_3^2y_2^2}{y_3 y_2+Zy_3y_1}\leq \frac{[(\tilde Y\tilde Y^\dagger)_{\gamma\gamma}(\tilde Y^\dagger \tilde Y)_{\delta\delta}]^\frac{1}{2}}{|\tilde Y_{\gamma\delta}|}
\leq \frac{(y_3y_2)^2+Z(y_3y_1)^2}{y_3y_2-Zy_3y_1}\,,
\ee
where $Z\equiv \frac{2y_2^2-y_1^2}{y_2^2-2y_1^2}$.
\end{theorem}
\begin{proof}
The proof is identical to Theorem \ref{thm:sv3}.
\end{proof}

The following Theorem defines an approximation $U'_L$ for the matrix $U_L$, and bounds its error in terms of the deviation of the matrix $U_L^\dagger U'^{}_L$ from the identity matrix.
\begin{theorem}
\label{thm:master0}
Let Y be a complex 3 by 3 matrix,  let $Y=U_L^{}\diag(y_1,y_2,y_3)U_R^\dagger$ be its SVD with $y_1\leq y_2\leq y_3$, and let $\tilde Y$ be the cofactor matrix of $Y$. Furthermore, let $\beta$ and $\delta$ be the indices of the longest columns of $Y$ and $\tilde Y$ respectively. Define the matrix $U_L'$ via
\be
U'_{Li3}\equiv \frac{Y_{i\beta}}{\bigl[(Y^\dagger Y)_{\beta\beta}\bigr]^\frac{1}{2}}\,,
\qquad
U'_{Li1}\equiv \frac{\tilde Y^*_{i\delta}}{\bigl[(\tilde Y^\dagger \tilde Y)_{\delta\delta}\bigr]^\frac{1}{2}}\,,
\qquad
U'_{Li2}\equiv \frac{(\epsilon_{ijk}U'_{L\,j3}U'_{L\,k1})^*}{\bigl[1-|U'^*_{L\,i3}U'^{}_{L\,i1}|^2\bigr]^\frac{1}{2}}\,.
\ee
 Finally, let $V_L\equiv U^\dagger_L U'^{}_L$.
Then the following inequalities hold:
\be
\begin{pmatrix}
1-|V_{L11}|^2	&	|V_{L12}|^2	&	|V_{L13}|^2\\
|V_{L21}|^2		&	1-|V_{L22}|^2&	|V_{L23}|^2\\
|V_{L31}|^2		&	|V_{L32}|^2	&	1-|V_{L33}|^2
\end{pmatrix}
\leq \begin{pmatrix}
Z(\frac{y_1}{y_2})^{2}	& 	Z'(\frac{y_1}{y_2})^{2}	&	X(\frac{y_1}{y_3})^{2}\\
Z(\frac{y_1}{y_2})^{2}	&	Z'(\frac{y_1}{y_2})^{2}+X'(\frac{y_2}{y_3})^{2}	&	X(\frac{y_2}{y_3})^{2}\\
Z(\frac{y_1}{y_3})^{2}	&	X'(\frac{y_2}{y_3})^{2}	&	X(\frac{y_2}{y_3})^{2}
\end{pmatrix}\,,
\label{eq:master0}
\ee
with $X\equiv \frac{2y_3^2-y_2^2}{y_3^2-2y_2^2}$, $Z\equiv \frac{2y_2^2-y_1^2}{y_2^2-2y_1^2}$, $X'\equiv (1-\frac{9y_1^2}{y_3^2})^{-1}X$, $Z'\equiv (1-\frac{9y_1^2}{y_3^2})^{-1}Z$.
\end{theorem}

\begin{proof}
Starting with the third column of eq.~(\ref{eq:master0}), we have 
\be
|V_{L\,i3}|^2
=\frac{y_i^2}{y_3^2}\frac
{\frac{| U_{R\beta i}|^2}{|U_{R\beta 3}|^2  }}
{1+\frac{| U_{R\beta 2}|^2}{|U_{R\beta 3}|^2}\frac{y_2^2}{y_3^2}+\frac{| U_{R\beta 1}|^2}{|U_{R\beta 3}|^2}\frac{y_1^2}{y_3^2}}\,,
\ee
where  we have used the SVD of $Y$.
For $i=1,2$ use Lemma \ref{lem:basic} to get
\be
|V_{Li3}|^2\leq X\frac{y_i^2}{y_3^2}\,.
\ee
For $i=3$, again by use of Lemma \ref{lem:basic},
\be
1-|V_{L33}|^2=\frac
{\frac{| U_{R\beta 2}|^2}{|U_{R\beta 3}|^2}\frac{y_2^2}{y_3^2}+\frac{| U_{R\beta 1}|^2}{|U_{R\beta 3}|^2}\frac{y_1^2}{y_3^2}}
{1+\frac{| U_{R\beta 2}|^2}{|U_{R\beta 3}|^2}\frac{y_2^2}{y_3^2}+\frac{| U_{R\beta 1}|^2}{|U_{R\beta 3}|^2}\frac{y_1^2}{y_3^2}}
\leq X\frac{y_2^2}{y_3^2}\,,
\ee
which completes the proof of the third column of eq.~(\ref{eq:master0}). To prove the inequalities of the first column, write it by use of the SVD of $\tilde Y$ as
\be
|V_{L\,i1}|^2
=\frac{y_1^{2}}{y_i^{2}}\frac
{\frac{| U_{R\delta i}|^2}{|U_{R\delta 1}|^2  }}
{1+\frac{| U_{R\delta 2}|^2}{|U_{R\delta 1}|^2}\frac{y_1^2}{y_2^2}+\frac{| U_{R\delta 3}|^2}{|U_{R\delta 1}|^2}\frac{y_1^2}{y_3^2}}\,.
\ee
Applying Lemma \ref{lem:basic2} and using identical reasoning as before we get the bounds in the first column of eq.~(\ref{eq:master0}).
Finally, let us consider the second column of $V_L$.
By defining $\epsilon\equiv (V_L^\dagger V^{}_L)_{31}=(U'^\dagger_L U_L'^{})_{31}$ we can conveniently rewrite 
\be
|V_{L\,i2}|^2=\frac{1-|V_{L\,i1}|^2-|V_{L\,i3}|^2+|V_{L\,i1}|^2|V_{L\,i3}|^2-|\epsilon - V^*_{L\,i3}V^{}_{L\,i1}|^2}{1-|\epsilon|^2}
\ee
This equation can be used to bound the $V_{L\,12}$ and $V_{L\,32}$ elements by use of the already verified bounds, as well as $|V_{L\,11}|\leq 1$ (which follows from unitarity of the first column of $V_L$).
\be
|V_{L\,12}|^2\leq 
\frac{1-|V_{L\,11}|^2-|V_{L\,13}|^2+|V_{L\,11}|^2|V_{L\,13}|^2}{1-|\epsilon|^2}
\leq
Z\frac{y_1^2}{y_2^2}\frac{1}{1-|\epsilon |^2}\,.
\label{eq:V12}
\ee
Simlarly, 
\be
|V_{L\,32}|^2\leq 
\frac{1-|V_{L\,31}|^2-|V_{L\,33}|^2+|V_{L\,31}|^2|V_{L\,33}|^2}{1-|\epsilon|^2}
\label{eq:V32}
\leq
X\frac{y_2^2}{y_3^2}\frac{1}{1-|\epsilon |^2}\,.
\ee
For $V_{L\,22}$, it is simplest to use unitarity of the second column of  $V_L$ as well as eq.~(\ref{eq:V12}) and (\ref{eq:V32}) to get
\be
1-|V_{L\,22}|^2=|V_{L\,12}|^2+|V_{L\,32}|^2\leq
\left(Z\frac{y_1^2}{y_2^2}+X\frac{y_2^2}{y_3^2}\right)\frac{1}{1-|\epsilon |^2}\,.
\ee
Finally, we must obtain an upper bound on $|\epsilon|$:
\be
|\epsilon|^2
=\frac{|(Y^\dagger \tilde Y^*)_{\beta\delta}|^2}{(Y^\dagger Y)_{\beta\beta}(\tilde Y^\dagger\tilde Y)_{\delta\delta} }
\leq\frac{(y_1y_2y_3)^2\delta_{\beta\delta}}
{\left(\frac{1}{3}y_3^2\right)\left(\frac{1}{3}y_3^2y_2^2 \right)}\,,
\ee
and hence
\be
|\epsilon|\leq 3\,\frac{y_1}{y_3}\,\delta_{\beta\delta}\,,
\ee
which completes the proof of  the Theorem.
\end{proof}


\begin{theorem}
Let Y be a complex 3 by 3 matrix,  let $Y=U_L^{}\diag(y_1,y_2,y_3)U_R^\dagger$ be its SVD with $y_1\leq y_2\leq y_3$, and let $\tilde Y$ be the cofactor matrix of $Y$. Furthermore, let $\alpha$ and $\gamma$ be the indices of the longest rows of $Y$ and $\tilde Y$ respectively. Define the matrix $U_R'$ via
\be
U'_{Ri3}\equiv \frac{Y^*_{\alpha i}}{\bigl[(YY^\dagger)_{\alpha\alpha}\bigr]^\frac{1}{2}}\,,
\qquad
U'_{Ri1}\equiv \frac{\tilde Y_{\gamma i}}{\bigl[(\tilde Y \tilde Y^\dagger)_{\gamma\gamma}\bigr]^\frac{1}{2}}\,,
\qquad
U'_{Ri2}\equiv \frac{(\epsilon_{ijk}U'_{R\,j3}U'_{R\,k1})^*}{\bigl[1-|U'^*_{R\,i3}U'^{}_{R\,i1}|^2\bigr]^\frac{1}{2}}\,,
\ee
Finally, let $V_R\equiv U^\dagger_R U'^{}_R$. Then the inequalities (\ref{eq:master0}) hold with the substitution $L\to R$.
\end{theorem}
\begin{proof}
The proof is identical to the one of Theorem \ref{thm:master0}.
\end{proof}

\bibliography{paper}

\bibliographystyle{JHEP}

\end{document}